\documentclass[preprint,12pt]{elsarticle}
\usepackage{etoolbox}
\makeatletter
\def\ps@pprintTitle{%
	\let\@oddhead\@empty
	\let\@evenhead\@empty
	\def\@oddfoot{\reset@font\hfil\thepage\hfil}
	\let\@evenfoot\@oddfoot
}
\makeatother
\usepackage{graphics}
\usepackage{amsmath, amsthm, amssymb}
\usepackage{natbib}
\usepackage[colorlinks=true]{hyperref}
\usepackage{lscape}
\usepackage[utf8]{inputenc} 
\usepackage{amsthm}
\usepackage{enumerate}
\usepackage{mathrsfs}
\usepackage{setspace}
\usepackage{multirow}
\usepackage{graphicx}
\usepackage{amsfonts}
\usepackage{verbatim}
\usepackage{amssymb}
\usepackage{amsthm}
\usepackage{caption}
\usepackage{subcaption}
\usepackage{multirow,bigstrut,bigdelim}
\newtheorem{theorem}{Theorem}[section]

\newtheorem{corollary}{Corollary}[section]
\theoremstyle{plain}
\newtheorem{definition}{Definition}[section]
\newtheorem{example}{Example}[section]

\theoremstyle{remark}

\numberwithin{equation}{section}
\usepackage[mathscr]{euscript}
\usepackage{geometry} 
\usepackage{graphicx,lscape} 
\usepackage{booktabs}
\usepackage{array}
\usepackage{paralist} 
\usepackage{verbatim}
\usepackage{makecell}
\biboptions{comma,round,authoryear}

\begin{document}
\begin{frontmatter}	
		\title{\textbf{Inaccuracy and divergence measures based on \\ survival extropy, their properties and applications in \\ testing and image analysis }}		
		\author{Saranya P.$ ^* $, S.M.Sunoj}
		\ead{smsunoj@cusat.ac.in, saranyapanat96@gmail.com}
		\cortext[cor1]{Corresponding author}
		\address{Department of Statistics,\\Cochin University of Science and Technology,\\Kerala, INDIA 682 022}
\begin{abstract}
This article introduces novel measures of inaccuracy and divergence based on survival extropy and their dynamic forms and explores their properties and applications. To address the drawbacks of asymmetry and range limitations, we introduce two measures: the survival extropy inaccuracy ratio and symmetric divergence measures. The inaccuracy ratio is utilized for the analysis and classification of images. A goodness-of-fit test for the uniform distribution is developed using the survival extropy divergence. Characterizations of the exponential distribution are derived using the dynamic survival extropy inaccuracy and divergence measures. The article also proposes non-parametric estimators for the divergence measures and conducts simulation studies to validate their performance. Finally, it demonstrates the application of symmetric survival extropy divergence in failure time data analysis.
\end{abstract}
\end{frontmatter}
\section{Introduction}
Consider an experimenter evaluating the probabilities for various possible outcomes of an experiment. Then the precision of finding the probability can be compromised either due to lack of sufficient information, or due to some of the information is incorrect. All statistical estimation and inference problems revolve around formulating probability statements that could be imprecise in either, or both, of these aspects. Based on this idea, \cite{kerridge1961inaccuracy} proposed a measure, known as the inaccuracy measure (or the Kerridge measure) between two probability functions. Let $X$ and $Y$ be two absolutely continuous random variables with cumulative distribution functions (CDFs) $F$ and $G$ and probability density functions (PDFs) $f$ and $g$ respectively. Then the Kerridge inaccuracy measure is defined as, $I(F, G) = -\int_{0}^{\infty} {f(x) \log g(x) dx}$. It can be defined as the sum of Shannon entropy (see \cite{shannon1948mathematical}) and Kullback-Leibler (KL) divergence (\cite{kullback1951information}) as, $I(F, G) = H(F) + K(F, G)$, where $H(F) = -\int_{0}^{\infty}{f(x) \log f(x) dx}$ is the Shannon entropy and $K(F, G) = \int_{0}^{\infty}{f(x) \log \frac{f(x)}{g(x)}dx}$ is the KL divergence. A key relationship between an inaccuracy measure and a measure of uncertainty (such as entropy or extropy) for a random variable is that the inaccuracy measure equals the corresponding uncertainty measure when the distributions are identical. {
Inaccuracy measures based on cumulative distribution function (CDF), replacing the PDF in $I(F, G)$ are also of great interest among researchers as they are in general, more stable since the distribution function is more regular because it is defined in an integral form unlike the density function defined as the derivative of the distribution (see \cite{rao2004cumulative}, \cite{taneja2012dynamic} and \cite{nair2022reliability}). 
Recently, the notion of extropy has been considered a new measure of uncertainty of a random variable defined as a complementary dual of the entropy by \cite{lad2015extropy}.
According to \cite{lad2015extropy}, the extropy of the random variable $X$ is defined as
\begin{equation}
    J(X)=-\frac{1}{2}\int_0^{\infty} f^2(x) dx.
\end{equation}
\cite{qiu2017extropy} derived the characterization outcomes and symmetric properties of the extropy of order statistics and record values. \cite{lad2015extropy} presented the properties of this measure, including the maximum extropy distribution and statistical applications. Also, \cite{qiu2018extropy} provided two estimators for the extropy of a continuous random variable. Extropy
and variational distance were compared by \cite{yang2019bounds}. They identified the distribution that, within a specified variation distance from any given probability distribution, achieves the minimum or maximum extropy among these distributions. A lifetime expression of the extropy of a mixed system was studied by \cite{qiu2019extropy}. Based on the survival function $\bar{F}$, \cite{jahanshahi2020cumulative} and \cite{abdul2021dynamic} defined and studied cumulative residual extropy (also called survival extropy) of $X$, 
\[
J_s(X)=-\frac{1}{2}\int_{0}^{\infty}\bar{F}^2(x)dx
\]
and dynamic survival extropy as $J_s(X;t)=-\frac{1}{2}\int_{t}^{\infty}\left(\frac{\bar{F}(x)}{\bar{F}(t)}\right)^2dx$.

\cite{jose2022renyi} provided the relevance of constructing ordered random variables from random sample in the context of Renyi entropy information contained in a random variable. \cite{hashempour2024dynamic} introduced an extropy based dynamic cumulative past inaccuracy measure. They studied the characterization problem and stochastic ordering for this measure. Also, \cite{mohammadi2022interval} proposed modified interval weighted cumulative residual and past extropies as generalized measure. Recently, \cite{hashempour2024new} provided a new measure of inaccuracy based
on extropy for record statistics between distributions of the nth upper (lower) record value
and parent random variable and discussed some properties. They defined the extropy-inaccuracy as follows.
\begin{equation}
    J(f,g)=-\frac{1}{2}\int_0^\infty f(x)g(x)dx.
\end{equation}
 Also, the discrimination information based
on extropy and inaccuracy between density functions $f(x)$ and $g(x)$ can be defined by 
\begin{equation}\label{eq:R3discrimeasuredensity}
    J(f|g)=\frac{1}{2}\int_0^\infty [f(x)-g(x)]f(x)dx.
\end{equation}
\cite{saranya2024relative} defined relative cumulative extropy, its dynamc forms, properties, and applications in testing uniformity. 
Extropy-based divergence measures are considerably less explored in the literature, and our focus is on developing effective discrimination measures based on extropy and its extensions. 
Motivated by (\ref{eq:R3discrimeasuredensity}), we derive an inaccuracy measure and a divergence measure based on survival functions of $X$ and $Y$, and explore its properties. The article is structured as follows: In Section 2, we introduce an inaccuracy measure called survival extropy inaccuracy and survival extropy inaccuracy ratio. Since the inaccuracy ratio is non-negative, it is used for the analysis and classification of images. In Section 3, we define survival extropy divergence and propose a test for goodness of fit for the uniform $U(0,b)$ distribution. Section 4 extends the defined measures to dynamic cases, specifically for residual lifetime random variables, and discusses their properties. Estimation and simulation studies are presented in Section 5. Application in real data is discussed in Section 6.
\section{Survival extropy inaccuracy and survival extropy inaccuracy ratio}
In this section, we introduce two new measures in information theory. The first one is the survival extropy inaccuracy measure which is given by the following definition.
\begin{definition}
Let $X$ and $Y$ be nonnegative continuous random variables with survival functions $\bar{F}$ and $\bar{G}$ respectively. Then the survival extropy inaccuracy (SEI) measure between the distributions is defined as
    \begin{equation}
\xi J_s(X,Y)=-\frac{1}{2}\int_0^\infty \bar{F}(x)\bar{G}(x)dx.
     \end{equation}
\end{definition}
$\xi J_s(X,Y)$ is always negative and is equal to $J_s(X)$ for $\bar{F}=\bar{G}$. This measure serves as a valuable tool for assessing errors in experimental results. The SEI quantifies the error or inaccuracy resulting from assuming 
$\bar{G}(x)$ instead of $\bar{F}(x)$. Additionally, this measure can be used to compare multiple approximations to determine which assumed model is closer to the true model.
Let $Y=aX+b$, $a>0$, $b\geq0$. We have $F_Y(x)=F_X(\frac{x-b}{a})$ implies 
\begin{equation*}
\xi J_s(X,Y)=-\frac{1}{2}\int_0^\infty \bar{F}(x)\bar{F}\left(\frac{x-b}{a}\right)dx.
     \end{equation*}
\begin{example}\label{Example:R3expinaccuracy}
    Let $X$ and $Y$ be two exponential random variables with survival functions  $e^{-{\lambda_1}x}$ and $e^{-{\lambda_2}x}$ respectively. Then, the survival extropy inaccuracy $\xi J_s(X,Y)=-\frac{1}{(\lambda_1+\lambda_2)}$. 
\end{example}
\cite{jahanshahi2020cumulative} derived the upper bound for cumulative residual extropy. Similarly, we derive a sufficient condition for survival extropy inaccuracy to be finite. 
\begin{theorem}
    Let $X$ and $Y$ be two non-negative random variables with finite variances. If for some $p>1/4$, $E[X^p]<+\infty$ and $E[Y^p]<+\infty$, then $\xi J(X,Y)\in (-\infty,0]$.
\end{theorem}
\begin{proof}

    \begin{equation*}
        \begin{split}
            \int_0^\infty \bar{F}(x)\bar{G}(x)dx&=\int_0^1 \bar{F}(x)\bar{G}(x)dx+\int_1^\infty \bar{F}(x)\bar{G}(x)dx\\
            &\leq 1+\int_1^\infty \left[\frac{E[X^p]E[Y^p]}{x^{2p} }\right]^2 dx\\
            &\leq 1+(E[X^p]E[Y^p])^2\int_1^\infty \frac{1}{x^{4p}}dx\\
            &\leq 1+\frac{(E[X^p]E[Y^p])^2}{(4p-1)} < \infty.
        \end{split}
    \end{equation*}
    We use Markov inequality in the third step. The last integral is finite if $p>(1/2)$. So, 
    \[
    -\frac{1}{2} \int_0^\infty \bar{F}(x)\bar{G}(x)dx>-\infty.
    \]
    Since $0\leq\bar{F}(x)\bar{G}(x)$, we have $-\bar{F}(x)\bar{G}(x)\leq 0$, which implies $-(1/2)\int_0^\infty \bar{F}(x)\bar{G}(x)\leq 0$, if it exists and is the upperbound. Thus the result follows.
\end{proof}
Since $\xi J_s(X,Y)$ is always negative, we define a nonnegative measure useful for measuring discrepancy between images. 
Motivated by \cite{kayal2020quantile}, we define an inaccuracy ratio called cumulative extropy inaccuracy ratio which can be applied to image analysis as follows.
\begin{definition}
    Let $X$ and $Y$ be two non-negative continuous random variables with survival functions $\bar{F}$ and $\bar{G}$ respectively. Then the survival extropy inaccuracy ratio between the distributions is defined as
    \begin{equation}
I \xi(X,Y)=\frac{\xi J_s(X,Y)}{J_s(X)}.
         \end{equation}  
 \end{definition}
 
 $I \xi(X,Y)$ is always positive and is equal to 1 for $\bar{F}=\bar{G}$. Also, it provides an adimensional
measure of closeness between $X$ and $Y$. Further, $I \xi(X,Y)$ is not symmetric since $I \xi(X,Y)\neq I \xi(Y,X)$. Note that $I \xi(X,Y)$ can be interpreted as the amount of dissimilarity carried by the survival extropy when the true survival function $\bar{F}$ is replaced by another survival function $\bar{G}$.
\begin{example}
 In the Example \ref{Example:R3expinaccuracy}, the survival extropy inaccuracy ratio is given by 
 \[
  I \xi (X,Y)=-\frac{{1/(\lambda_1+\lambda_2)}}{-1/4\lambda_1}={\frac{2\lambda_1}{\lambda_1+\lambda_2}}
 \]
\end{example}
\begin{theorem}
     If $\bar{F}(x)\geq \bar{G}(x)$, then $I \xi(X,Y)\leq I \xi(Y,X)$.
 \end{theorem}
 \begin{proof}
     If $\bar{F}(x)\geq \bar{G}(x)$, then $\bar{F}^2(x)\geq \bar{G}^2(x)$, which implies $J_s(X)\geq J_s(Y)$ and we get $I \xi(X,Y)\leq I \xi(Y,X)$. 
 \end{proof}

\subsection{Classification of images using survival extropy inaccuracy ratio}

In this section, we illustrate the usefulness of the survival extropy inaccuracy ratio for classification of image data sets.  We use chineese MNIST data by \cite{kaggle_chinese_mnist} for our analysis. We consider images of five numbers in the Chinese scripts corresponding to 0, 1, 2, 3, and 4. For convenience, we denote the images corresponding to each of the number scripts by $\bold{0}$, $\bold{1}$, $\bold{2}$, $\bold{3}$, and $\bold{4}$ respectively. For analyzing the images, we first converted each
image into $28\times 28$ pixel data with 784 cells. The grey level of each cell is measured by a real number ranging from 0 (black) to 1 (white). 

Figure \ref{Fig:ChineeseMNISTimage} represents a sample of 6 images of each of the first four natural numbers in Chinese scripts, 0, 1, 2, 3 and 4. The first column of the figure represents $\bold{0}$, the second column is $\bold{1}$, and so on. We have taken 500 samples of images of $\bold{0}$, $\bold{1}$, $\bold{2}$, $\bold{3}$, and $\bold{4}$, converted them into numerical data, ranging from 0 to 1 and estimated $I \xi(X,Y)$. We estimate the survival extropy inaccuracy for each pair of images such as $(\bold{0}, \bold{0}), (\bold{0}, \bold{1}), (\bold{0}, \bold{2}), (\bold{0}, \bold{3})$ and $(\bold{0}, \bold{4})$, and determine the discrepancy of images $\bold{0}, \bold{1}, \bold{2}, \bold{3}, \bold{4}$ from $\bold{0}$, based on $\hat{I} \xi(X,Y)$. The closer the estimate is to 1, the more similar the images are. This means that an estimate approaching 1 suggests a high degree of resemblance between the images being compared.  Table \ref{Table:Inaccuracy ratio chineeseMNIST} shows the values of estimator $\hat{I}\xi(X,Y)$ of a random image $\bold{0}$ with $\bold{0}$, $\bold{1}$, $\bold{2}$, $\bold{3}$ and $\bold{4}$ respectively. The estimator $\hat{I}\xi(X,Y)$ is defined as follows. 
\begin{figure}[h!]
\centering
\includegraphics[scale=1.5]{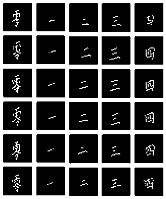}
\caption{Chineese MNIST image sample}
\label{Fig:ChineeseMNISTimage}
\end{figure}
  Let $X_1, X_2, X_3, \cdots, X_n$ be a random sample drawn from a population having survival function $\bar{F}$ and $Y_1, Y_2, Y_3, \cdots, Y_n$ be a random sample drawn from a population having survival function $\bar{G}$. Let $X_{(i)}$ and $Y_{(i)}$ be the order statistic of $X$ and $Y$ respectively. Then the empirical plug-in estimator for $I \xi(X,Y)$ is given by
\[
 \hat{I} \xi(X,Y)=\frac{\sum_{i=1}^{n-1}P(X\geq X_{(i)})P(Y\geq X_{(i)})(X_{(i+1)}-X_{(i)})}{\sum_{i=1}^{n-1}(P(X\geq X_{(i)}))^2 (X_{(i+1)}-X_{(i)})}.
\]
\begin{table}[h!]
\centering
\begin{tabular}{|c|c|c|c|c|c|}
\hline
&\includegraphics[width=0.05\textwidth]{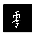}0&\includegraphics[width=0.04\textwidth]{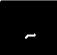}1&\includegraphics[width=0.045\textwidth]{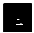}2& \includegraphics[width=0.04\textwidth]{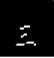}3 &\includegraphics[width=0.05\textwidth]{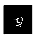}4 \\
\hline
\includegraphics[width=0.05\textwidth]{Zero.png}0 & 1 & 0.7194&0.7498&0.7676& 0.7967\\
\hline
\end{tabular}
\caption{Inaccuracy ratio estimate of images}
\label{Table:Inaccuracy ratio chineeseMNIST}
\end{table}
\begin{figure}{h!}
    \centering
    \includegraphics[width=0.5\textwidth]{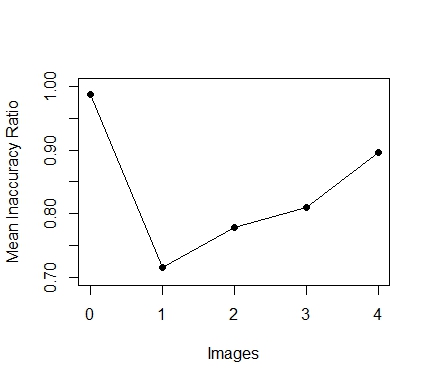}
    \caption{Mean Inaccuracy ratio estimates}
    \label{Fig:R3Mean Inaccuracy ratio}
\end{figure}
 It shows that $\hat{I}\xi(X,Y)$ between the same $\bold{0}$ and $\bold{0}$ is 1 and between $\bold{0}$ and others are deviated from 1 in an increasing pattern.
Now, we estimate the inaccuracy ratio of all combinations of 500 images of 0 with other numbers 1, 2, 3, and 4, and calculate the mean of estimates which are given in Table \ref{Table:Mean of inaccuracy Chinese MNIST}. The ratio is close to 1 for $\bold{0}$, 0.7163 for $\bold{1}$ and increases for $\bold{2},\bold{3}$ and $\bold{4}$. Figure \ref{Fig:R3Mean Inaccuracy ratio} shows the trend of the mean of inaccuracy ratio between $(\bold{0}$ with all other images $\bold{1}, \bold{2}, \bold{3}$ and $\bold{4}$.

\begin{table}[h!]
\centering
\begin{tabular}{|c|c|c|c|c|c|}
\hline
&$\bold{0}$&$\bold{1}$&$\bold{2}$& $\bold{3}$ &$\bold{4}$ \\
\hline
$\bold{0}$ & 0.9865 & 0.7163 &0.7784& 0.8093& 0.8957\\

\hline
\end{tabular}
\caption{Mean of Inaccuracy ratio estimates of 500 images}
\label{Table:Mean of inaccuracy Chinese MNIST}
\end{table}
Next, our attempt is to discriminate and classify the images using the survival extropy inaccuracy ratio. First, we divide the set of images into two groups; the training set and the test set. We have 500 samples of each number 0, 1, 2, 3 and 4. The ratio of division is taken as 3:2 such that the training set of each number consists of 300 images and the test set consists of 200 images. From Table \ref{Table:Mean of inaccuracy Chinese MNIST}, one can find an increasing trend in the mean inaccuracy ratio from $(\bold{0},\bold{1})$ to the inaccuracy of $(\bold{0},\bold{4})$, and thus propose $\hat{I}\xi(\bold{0},\bold{A})<\hat{I}\xi(\bold{0},\bold{B}), \; A, B = 1,2, 3; \; A < B $ as a general criterion for classifying an image from a pair of images in the test set.  Now using the training set, we obtain the values of inaccuracy ratios of 300 $\times$ 300 = 90000 pairs for each images and Table \ref{Table:R3Probability of different events} provides the probability corresponding to each event. The probability of $\hat{I}\xi(\bold{0},\bold{1})<\hat{I}\xi(\bold{0},\bold{2})$ is 0.99, the probability of the event $\hat{I}\xi(\bold{0},\bold{1})<\hat{I}\xi(\bold{0},\bold{3})$ is 1, and the probability of $\hat{I}\xi(\bold{0},\bold{1})<\hat{I}\xi(\bold{0},\bold{4})$ is 1. Next, we try to classify the test set based on the above criterion. Let $n_1$, $n_2$, $n_3$ and $n_4$ denote the sample size of $\bold{1}$, $\bold{2}$, $\bold{3}$ and $\bold{4}$ respectively. Let $N$ denote the total number of images selected for each classification. The Table \ref{Table:R3Accuracy of classification of 1&2} gives the accuracy of classification of images to $\bold{1}$ and $\bold{2}$. 
\begin{table}[h!]
\centering
\begin{tabular}{|c|c|c|}
\hline
Event& No. of outcomes of the event & Probability of the event\\
\hline
$\hat{I}\xi(\bold{0},\bold{1})<\hat{I}\xi(\bold{0},\bold{2})$&89970&0.9997\\
$\hat{I}\xi(\bold{0},\bold{1})<\hat{I}\xi(\bold{0},\bold{3})$& 90000&1 \\
$\hat{I}\xi(\bold{0},\bold{1})<\hat{I}\xi(\bold{0},\bold{4})$&90000 &1\\
$\hat{I}\xi(\bold{0},\bold{2})<\hat{I}\xi(\bold{0},\bold{3})$&50400&0.56\\ $\hat{I}\xi(\bold{0},\bold{2})<\hat{I}\xi(\bold{0},\bold{4})$&75000&0.8333\\
$\hat{I}\xi(\bold{0},\bold{3})<\hat{I}\xi(\bold{0},\bold{4})$&63300&0.70333 \\
\hline
\end{tabular}
\caption{Probability of different events based on training set}
\label{Table:R3Probability of different events}
\end{table}
Considering that the inaccuracy between $\bold{0}$ and $\bold{1}$ is less than the inaccuracy between $\bold{0}$ and $\bold{2}$ and likewise, the classification is made based on the comparison of estimate values. Specifically, for a pair of images $(\bold{A},\bold{B})$, if the inaccuracy ratio between $\bold{0}$ and $\bold{A}$ is less than the inaccuracy ratio between $\bold{0}$ and $\bold{B}$ then the $\bold{A}$ is classified into group $\bold{1}$ and $\bold{B}$ is classified into group $\bold{2}$, and check with their original groups to verify whether the classification is true. Tables \ref{Table:R3Accuracy of classification of 1&3} and \ref{Table:R3Accuracy of classification of 1&4} show the accuracy of classification of $\bold{1}$ \& $\bold{3}$ and $\bold{1}$ \& $\bold{4}$ respectively.
\begin{table}[h!]
\centering
\begin{tabular}{|c|c|c|c|c|}
\hline
$(n_1,n_2)$& $N$& Truely classified as $\bold{1}$ &Truely classified as $\bold{2}$ & Accuracy\\
\hline
$(25,25)$ & 50 & 25 &25 & 1\\ 
\hline
$(50,50)$ &100 & 50 &50 &1\\
\hline
$(100,100)$ &200 & 97 & 97 & 0.97\\
\hline
$(200,200)$ &400 &198 &198&0.99\\
\hline
\end{tabular}
\caption{Accuracy of classification of $\bold{1}$ and $\bold{2}$.}
\label{Table:R3Accuracy of classification of 1&2}
\end{table}

\begin{table}[h!]
\centering
\begin{tabular}{|c|c|c|c|c|}
\hline
$(n_1,n_2)$& $N$& Truely classified as $\bold{1}$ &Truely classified as $\bold{3}$ & Accuracy\\
\hline
$(25,25)$ & 50 & 25 &25 & 1\\ 
\hline
$(50,50)$ &100 & 50 &50 &1\\
\hline
$(100,100)$ &200 & 99 & 99 & 0.99\\
\hline
$(200,200)$ &400 &198 &198&0.99\\
\hline
\end{tabular}
\caption{Accuracy of classification of $\bold{1}$ and $\bold{3}$.}
\label{Table:R3Accuracy of classification of 1&3}
\end{table}
\begin{table}[h!]
\centering
\begin{tabular}{|c|c|c|c|c|}
\hline
$(n_1,n_2)$& $N$& Truely classified as $\bold{1}$ &Truely classified as $\bold{4}$ & Accuracy\\
\hline
$(25,25)$ & 50 & 25 &25 & 1\\ 
\hline
$(50,50)$ &100 & 50 &50 &1\\
\hline
$(100,100)$ &200 & 100 & 100 & 1\\
\hline
$(200,200)$ &400 &199 &199&0.995\\
\hline
\end{tabular}
\caption{Accuracy of classification of $\bold{1}$ and $\bold{4}$.}
\label{Table:R3Accuracy of classification of 1&4}
\end{table}

\section{Survival extropy divergence}
\begin{definition}
 Let $X$ and $Y$ be non-negative continuous random variables with survival functions $\bar{F}$ and $\bar{G}$ respectively. Then the survival extropy divergence measure between the distributions is defined as
    \begin{equation}
SJ(\bar{F}|\bar{G})=\frac{1}{2}\int_0^\infty (\bar{F}(x)-\bar{G}(x))\bar{F}(x)dx.
         \end{equation}   
\end{definition}
We have
\[
SJ(\bar{F}|\bar{G})=\xi J_s(X,Y)-J_s(X).
\]
\begin{example}
    Let $X$ and $Y$ follows exponential distribution with mean $1/\lambda_1$ and $1/\lambda_2$ respectively. We have the survival extropy divergence of $X$ and $Y$ as
    \[
    SJ(\bar{F}|\bar{G})=\frac{1}{2(\lambda_1+\lambda_2)}-\frac{1}{4\lambda_1}
    \]
    and extropy divergence as
    \[
    J(f|g)=\frac{\lambda_1\lambda_2}{2(\lambda_1+\lambda_2)}-\frac{\lambda_1}{4}
    \]
    Figure \ref{fig:R3survivalextropy} presents plots illustrating the extropy divergence and survival extropy divergence for various values of $\lambda_2$ at $\lambda_1=1$ and $\lambda_1=2$ in Figure \ref{fig:R3a} and Figure \ref{fig:R3b} respectively. These plots reveal distinct monotonic behaviors: $J(f|g)$ exhibits non-decreasing trends, while  $SJ(\bar{F},\bar{G})$ shows non-increasing trends. Specifically, for $\lambda_2=1$ in Figure \ref{fig:R3a} and $\lambda_2=2$ in Figure \ref{fig:R3b}, both divergences are zero, indicating identical distributions.
    \begin{figure}[h!]
  \centering
  \begin{subfigure}[b]{0.45\textwidth}
    \centering
    \includegraphics[width=\textwidth]{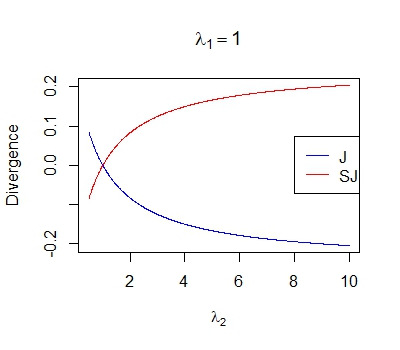}
    \caption{$J(f|g)$ and $SJ(\bar{F},\bar{G})$  for $\lambda_1=1$.}
    \label{fig:R3a}
  \end{subfigure}
  \hfill
  \begin{subfigure}[b]{0.45\textwidth}
    \centering
    \includegraphics[width=\textwidth]{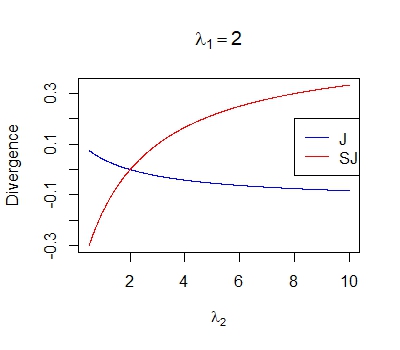}
    \caption{$J(f|g)$ and $SJ(\bar{F},\bar{G})$ for $\lambda_1=2$.}
    \label{fig:R3b}
  \end{subfigure}
  \caption{Extropy divergence and survival extropy divergence of two exponential distributions.}
  \label{fig:R3survivalextropy}
\end{figure}
\end{example}
In the following theorems we prove some properties of survival extropy divergence.
\begin{theorem}
    If $\bar{F}(x)\leq(\geq)\bar{G}(x)$, for every $x$, then $SJ(\bar{F}|\bar{G})\leq(\geq)0$.
\end{theorem}
\begin{proof}
    If $\bar{F}(x) \leq (\geq) \bar{G}(x)$, then $\bar{F}^2(x) \leq (\geq) \bar{F}(x)\bar{G}(x)$. Integrating both sides,
    \[
    \int_0^\infty \bar{F}^2(x)dx \leq (\geq) \int_0^\infty\bar{F}(x)\bar{G}(x)dx.
    \]
    which implies $SJ(\bar{F}|\bar{G})\leq (\geq) 0$.
\end{proof}
\begin{theorem}
 Let $X$ and $Y$ be two continous non-negative random variables with survival functions $\bar{F}$ and $\bar{G}$ respectively. The sum of survival extropy divergence between $(\bar{F}+\bar{G})/2$ and $\bar{G}$, and  between $(\bar{F}+\bar{G})/2$ and $\bar{F}$ is always equal to zero. That is,
 \begin{equation}
   SJ\left(\frac{\bar{F}+\bar{G}}{2}|\bar{G}\right)+SJ\left(\frac{\bar{F}+\bar{G}}{2}|\bar{F}\right)=0.  
 \end{equation}
\end{theorem}
\begin{proof}
    Since,
    \begin{equation}\label{Eq:R3sumofdivergence1}
    \begin{split}
      SJ\left(\frac{\bar{F}+\bar{G}}{2}|\bar{G}\right)&=\frac{1}{2}\int_0^\infty \left(\frac{\bar{F}(x)-\bar{G}(x)}{2}\right)\left(\frac{\bar{F}(x)+\bar{G}(x)}{2}\right)dx \\
      &=\frac{1}{2}\int_0^\infty \frac{\bar{F}^2(x)-\bar{G}^2(x)}dx=\frac{1}{4}(J_s(Y)-J_s(X)).
      \end{split}
    \end{equation}
 Similarly, we have
 \begin{equation}\label{Eq:R3sumofdivergence2}
 SJ\left(\frac{\bar{F}+\bar{G}}{2}|\bar{F}\right)=\frac{1}{4}(J_s(X)-J_s(Y)).
 \end{equation}
 From (\ref{Eq:R3sumofdivergence1}) and (\ref{Eq:R3sumofdivergence2}), we get
 \[
 SJ\left(\frac{\bar{F}+\bar{G}}{2}|\bar{G}\right)=-SJ\left(\frac{\bar{F}+\bar{G}}{2}|\bar{F}\right).
 \]
\end{proof}
\begin{theorem}\label{Theorem:R3survextdivbetF+GandF}
The survival extropy divergence between $\bar{F}$ and $\bar{G}$ is equal to twice the survival extropy divergence between $\bar{F}$ and $(\bar{F}+\bar{G})/2$.
  \begin{equation*}
      SJ(\bar{F}|\bar{G})=2\times SJ\left(\bar{F}|\frac{\bar{F}+\bar{G}}{2}\right).
  \end{equation*}  
\end{theorem}
\begin{corollary}
\begin{equation*}
  SJ \left(\frac{\bar{F}+\bar{G}}{2}|\bar{F} \right)+SJ\left(\frac{\bar{F}+\bar{G}}{2}|{\bar{G}}\right) = \frac{1}{2} (SJ(\bar{F}|\bar{G})+SJ(\bar{G}|\bar{F})).
\end{equation*}
   
\end{corollary}

The asymmetric nature of survival extropy divergence measure gives unequal values for $SJ(\bar{F}|\bar{G})$ and $SJ(\bar{G},\bar{F})$. It depends on the order of comparison, which makes them harder to interpret and understand. \cite{taneja2005generalized} introduced generalized symmetric divergence measures based on some asymmetric divergence measures. Symmetric measures treat both distributions equally which ensures a balance. This can be particularly important in comparisons where no single entity should be given precedence. Let $P={\{p_1,p_2,...,p_n\}}$ and $Q=\{q_1,q_2,...,q_n\}$ be two complete finite discrete probability distributions. Then the Jensen Shannon divergence (see \cite{sibson1969information}, \cite{burbea1982entropy}, \cite{burbea1982convexity}) between $P$ and $Q$ is given by 
\[
I(P||Q)=\frac{1}{2}\left( \sum_{i=1}^n p_i \ln \frac{2p_i}{p_i+q_i}+\sum_{i=1}^n q_i \ln \frac{2q_i}{p_i+q_i}\right)
\]
Using Kullback-Leibler divergence (\cite{kullback1951information}), it can be written as 
\[
I(P||Q)=\frac{1}{2}\left(K\left(P||\frac{P+Q}{2}\right) +K\left(Q||\frac{P+Q}{2}\right) \right).
\]
Motivated by this, we define a symmetric survival extropy divergence as follows.
\begin{definition}
  Let $X$ and $Y$ be two non-negative continuous random variables with survival functions $\bar{F}$ and $\bar{G}$ respectively. Let $SJ(\bar{F}|\bar{G})$ and $SJ(\bar{G}|\bar{F})$ be the corresponding survival extropy divergence measures. Then the symmetric survival extropy divergence measure between the distributions is defined as
    \begin{equation}\label{Equation:R3symdiv}
SSJ(\bar{F},\bar{G})=\frac{1}{2}\left(SJ\left(\bar{F}|\frac{\bar{F}+\bar{G}}{2}\right)+SJ\left(\bar{G}|\frac{\bar{F}+\bar{G}}{2}\right)\right).
    \end{equation}     
\end{definition}
From Theorem \ref{Theorem:R3survextdivbetF+GandF}, we have 
\begin{equation}\label{Equation:R3surextdivpart1}
SJ\left(\bar{F}|\frac{\bar{F}+\bar{G}}{2}\right)=\frac{1}{4}SJ(\bar{F}|\bar{G})
\end{equation}
and
\begin{equation}\label{Equation:R3surextdivpart2}
SJ\left(\bar{G}|\frac{\bar{F}+\bar{G}}{2}\right)=\frac{1}{4}SJ(\bar{G}|\bar{F}).
\end{equation}
Substituting (\ref{Equation:R3surextdivpart1}) and (\ref{Equation:R3surextdivpart2}) in (\ref{Equation:R3symdiv}), we rewrite the symmetric survival extropy divergence, $SSJ(\bar{F},\bar{G})$ as
\begin{equation}
 SSJ(\bar{F},\bar{G})=\frac{1}{8}\left( SJ(\bar{F}|\bar{G})+SJ(\bar{G}|\bar{F})\right).   
\end{equation}
 We have 
\[
SSJ(\bar{F},\bar{G})=SSJ(\bar{G},\bar{F}).
\]
Next, we develop a test for goodness of fit for uniform, $U(0,b)$ distribution using $SJ(\bar{F},\bar{G})$.
\subsection{Testing goodness of fit for uniform distribution}
Divergence measures give the discrepancy between two distributions. In this section, we used the idea of finding the discrepancy between an assumed distribution and an estimated one to define a test for goodness of fit for uniform $U(0,b)$ distribution. An important property of uniform distribution is that it obtains the maximum cumulative extropy among all distributions that possess a distribution function $F$ and have a given support on $(0, b)$.

\par Let $X_1, X_2,..., X_n$ be non-negative, independent, and identically distributed random variables from an absolutely continuous cdf $F$ with order statistics, $X_{(1)}\leq X_{(2)},...\leq X_{(n)}$. We choose $b=2E[X]$. Let $F_0(x)=x/b, \; 0\leq x\leq b$, denote the  uniform distribution $U(0,b)$.
Consider the testing problem with the hypothesis
\[
H_0: F(x)=F_0(x),\ \   vs.\ \   H_1: F(x)\neq F_0(x).
\]
Under the null hypothesis, $J(\bar{F}|\bar{F_0})=0$, and a large value of $J(\bar{F}|\bar{F_0})$ leads us to reject the null hypothesis $H_0$. To compute $SJ(\bar{F}|\bar{F}_0)$, complete information about $F$ and $F_0$ is required which is unavailable in testing problems, we develop a test statistic based on it.

Since, $J_s(\bar{F}_0)=-\frac{b}{6}$ and $\bar{F}_0=1-(x/b)$, we have
\begin{equation*}
    \begin{split}
         J(\bar{F}_0|\bar{F})&=-J_s(\bar{F}_0)-\int_0^\infty \bar{F}(x) \left(1-\frac{x}{b}\right)dx\\
        &=\frac{b}{6}-\int_0^\infty\bar{F}(x)+\frac{1}{2b}\int_0^\infty x\bar{F}(x)dx\\
        &=\frac{b}{6}-\frac{E[X]}{2}+\frac{E[X^2]}{4b}.
    \end{split}
\end{equation*}
Under $H_0$, $E[X]=\frac{b}{2}$, $E[X^2]=\frac{b^2}{3}$ and $J_s(\bar{F})=-\frac{b}{6}$. Replacing $b$ by $2\bar{x}$, $E[X]$ by $\bar{x}$ and $E[X^2]$ by $\frac{1}{n}\sum_{i=1}^{n}X_{i}^2$, we get a plug-in estimator of $ J(\bar{F}|\bar{F}_0)$ which is used as our test statistic 
\[
T_n=\frac{\sum_{i=1}^{n}{x_i}^2}{8n\bar{x}}-\frac{\bar{x}}{6}.
\]
We have $\frac{1}{n}\sum_{i=1}^{n}X_{i}$ and $\frac{1}{n}\sum_{i=1}^{n-1}X_i^2$ which are consistent estimators of $E[X]$ and $E[X^2]$ respectively. Evidently, $T_n$ is a consistent estimator of $SJ(\bar{F}_0|\bar{F})$ and under $H_0$, $T_n\rightarrow 0$ in probability. The $T_n$ is used as the test statistic such that we reject $H_0$ at the level of significance $\alpha$ if $T_n\geq T_{n,(1-\alpha)}$, where $T_{n,(1-\alpha)}$ is the $100(1-\alpha)$ percentile of $T_n$ under $H_0$. Since the distribution of $T_n$ under the null hypothesis has not been obtained analytically, the percentage point $T_{n,(1-\alpha)}$ is determined by Monte Carlo simulations.  We have generated 10,00000 samples of size $n$ and calculated the critical values for each $\alpha$. Table \ref{Table:CriticalvaluesofgoodnessoffitU(0,1)} gives the critical values of $U(0,1)$ distribution for different values of $n$ and $\alpha$ and Table \ref{Table:Powercomparisonuniformitywithbeta} gives the power of the test when the alternative distribution is considered as beta distribution.\\
 \begin{table}[h!]
	\begin{center}
		\caption{Critical values of the test statistic $T_n$ for $U(0,1)$}
		\label{Table:CriticalvaluesofgoodnessoffitU(0,1)} 		
 \scalebox{1.2}{ \begin{tabular}{p{0.5cm} p{2.2cm}  p{2.3cm} p{0.5cm} p{2.2cm} p{2.2cm}}
			\hline
                & & \ \ \ \ \ \ \ \ \ \ \ \ $T_n$  & & \\
                \hline
$n$ & $\alpha=0.01$ & $\alpha=0.05$ &$n$ & $\alpha=0.01$ & $\alpha=0.05$ \\
\hline
10 & 0.02166838 & 0.0144495& 28 & 0.01317029 & 0.009024229\\
11 & 0.02081598 & 0.01385329 &29 & 0.01294684 & 0.00880072 \\
12 & 0.01988166 & 0.01333844 &30  & 0.01269727 & 0.008714309\\
13 & 0.0191062  & 0.01288266 &31 & 0.0125022  & 0.008541361 \\
14 & 0.01859912 & 0.01244292 & 32 & 0.01239294 & 0.008475774 \\
15 & 0.0180311  & 0.01210819 &33 & 0.01222981 & 0.008400736 \\
16 & 0.01745607 & 0.01171947 &34  & 0.01205469 & 0.008233354 \\
17 & 0.01688279 & 0.01142615 &35  & 0.01184186 & 0.008175282 \\
18 & 0.01630101 & 0.0111104  &36 & 0.01160249 & 0.008006654 \\
19 & 0.01598378 & 0.01082509 &37& 0.01142032 & 0.00788299 \\
20 & 0.01550527 & 0.01053819 &38 & 0.01139521 & 0.007763844 \\
21 & 0.01521614 & 0.01030295&39 & 0.01111382 & 0.00765478  \\
22 & 0.01485849 & 0.0101308 &40  & 0.01108748 & 0.007624532 \\  
23 & 0.01448316 & 0.009873566&45&0.01045511&0.007266659 \\
24  & 0.01407422 & 0.009723266&50& 0.009864687 & 0.006856017 \\
25  & 0.01388488 & 0.009479421&55& 0.009403397 & 0.006559372 \\
26  & 0.01350353 & 0.009340834&60& 0.009110397 & 0.006312815 \\
27  & 0.01340029 & 0.009223798&100& 0.007045608 & 0.004870007 \\
			\hline
		\end{tabular}}  
	\end{center}
\end{table}	
\begin{table}[h]
    \begin{center}
    \caption{Power estimates for alternative beta distribution for different sample sizes and $\alpha$.}
		\label{Table:Powercomparisonuniformitywithbeta} 
    \begin{tabular}{ |p{3cm}| p{1.5cm} p{1cm}| p{1cm} p{1cm}| p{1cm} p{1cm}| p{1cm} p{1cm}}
    \hline
    \multicolumn{1}{|c|}{}&\multicolumn{2}{c|}{$n=10$} & \multicolumn{2}{c|}{$n=20$} &\multicolumn{2}{c|}
    {$n=40$}\\
    \hline
Distribution & {0.01} & {0.05} & {0.01} & {0.05}&{0.01} & {0.05} \\
         \hline
         $\beta(0.5,1)$& 0.305&0.551&0.619&0.832&0.919&0.981\\
         $\beta(0.1,0.5)$&0.869&0.924&0.992&0.997&1&1\\
         $\beta(0.3,0.5)$&0.585&0.770&0.858&0.941&0.987&0.997\\
         \hline
          \end{tabular}
    \end{center}
\end{table}
Next, we compare the powers of our test statistic for sample sizes 10, 20, and 40 with some well-known test statistics under the same alternatives for $\alpha=0.01$ and $\alpha=0.05$. We consider the Kolmogorov-Smirnov statistic (see \cite{an1933sulla} and \cite{smirnov1939estimate}), Anderson-Darling statistic (see \cite{anderson1954test}), Cramer-von Mises statistic (see \cite{cramer1928composition} and \cite{margenau1932richard}), Zamanzade statistic (see \cite{zamanzade2015testing}), and an extropy based test proposed by \cite{qiu2018extropy}.
\begin{itemize}
    \item[(i)]   Kolmogorov-Smirnov statistic
    \begin{equation*}
        KS=\max\left(\max_{1\leq i\leq n}\left(\frac{i}{n}-X_{(i)}\right), \max_{1\leq i\leq n}\left(X_{(i)}-\frac{i-1}{n}\right)\right).
    \end{equation*}
    \item [(ii)] Anderson-Darling statistic
    \begin{equation*}
        AD=-\frac{2}{n}\sum_{i=1}^n\left(\left(i-\frac{1}{2}\right)\log X_{(i)}+\left(n-i+\frac{1}{2}\right)\log(1-X_{(i)}\right)-n.
    \end{equation*}
    \item [(iii)] Cramer-von Mises statistic
    \begin{equation*}
        CM=\sum_{i=1}^n\left(X_{(i)}-\frac{2i-1}{2n}\right)+\frac{1}{12n}.
    \end{equation*}
    \item [(iv)] Zamanzade statistic
    \begin{equation*}
        TB=\sum_{i=1}^n \log\left(\frac{X_{(i+m)}-X_{(i-m)}}{\hat{F_n}(X_{(i+m)})-\hat{F_n}(X_{(i-m)})}\right)\left(\frac{\hat{F_n}(X_{(i+m)})-\hat{F_n}(X_{(i-m)})}{\sum_{i=1}^n(\hat{F_n}(X_{(i+m)})-\hat{F_n}(X_{(i-m)}))}\right),
    \end{equation*}
    where $m$ is the window size and $\hat{F_n}$ is the cdf estimator given by 
    \begin{equation*}
       \hat{F_n}(X_{(i)})=\frac{1}{n+2}\left( i+\frac{X_{(i)}-X_{(i-1)}}{X_{(i+1)}-X_{(i-1)}}\right), i=1,2,...,n. 
    \end{equation*}
    \item [(v)] Extropy based test statistic
    \begin{equation*}
        TU=\frac{1}{2n}\sum_{i=1}^n\frac{c_im/n}{X_{(i+m)}-X_{(i-m)}},
    \end{equation*}
    where $m$ is the window size and 
    \begin{equation*}
    c_i=    \begin{cases}
  1+\frac{i-1}{m}, & \mbox{  $ 1\leq i \leq m $,}\\
  2, & \mbox{  $m+1<i<n-m$,}\\
  1+\frac{n-i}{m}, & \mbox{  $n-m+1\leq i \leq n.$}
  \end{cases}
    \end{equation*}
\end{itemize}
The following probability distribution is also used as alternatives in testing uniformity, due to  \cite{dudewicz1981entropy} and \cite{zamanzade2015testing}, respectively given by 

   $C_k$: $k=1.5, 2$
    \begin{equation*}
    F(x)= \begin{cases}
  0.5-2^{k-1}(0.5-x)^k, & \mbox{  $ 0\leq x \leq 0.5 $,} \\
  0.5+2^{k-1}(x-0.5)^k, & \mbox{  $0.5<x\leq 1$.} 
  \end{cases}
    \end{equation*}

\begin{table}[]
    \begin{center}
    \caption{Power estimates of uniformity for different tests for different sample sizes and $\alpha$.}
		\label{Table:Powercomparisonuniformity} 
    \begin{tabular}{|p{1cm}| p{1cm}| p{1cm} p{1cm}| p{1cm} p{1cm}|}
    \hline
    \multicolumn{1}{|c|}{} & \multicolumn{1}{c|}{} &\multicolumn{2}{c|}
    {$C_{1.5}$}&\multicolumn{2}{c|}{$C_2$}\\
    \hline
$n$ & Test &{0.01} & {0.05}&{0.01} & {0.05} \\
         \hline
         10 & KS & 0.032&0.113&0.067&0.199\\
         &AD&0.033&0.121&0.072&0.227\\
         &CM&0.026&0.098&0.047&0.155\\
         &TB&0.005&0.028&0.009&0.035\\
         &TU&0.023&0.089&0.040&0.151\\
         &$T_n$& \textbf{0.058}&\textbf{0.186}&\textbf{0.165}&\textbf{0.368}\\
         \hline
         20&KS&0.048&0.150&0.124&0.315\\
         &AD&0.045&0.156&0.131&0.371\\
         &CM&0.031&0.119&0.072&0.250\\
         &TB&0.006&0.029&0.014&0.057\\
         &TU&0.042&0.155&0.131&0.363\\
         &$T_n$&\textbf{0.106}&\textbf{0.284}&\textbf{0.317}&\textbf{0.556}\\
         \hline
         40&KS&0.078&0.230&0.260&0.561\\
         &AD&0.065&0.241&0.330&0.695\\
         &CM&0.045&0.175&0.176&0.557\\
         &TB&0.013&0.066&0.085&0.290\\
         &TU&0.129&0.317&0.511&0.782\\
         &$T_n$&\textbf{0.202}&\textbf{0.434}&\textbf{0.568}&\textbf{0.787}\\
         \hline
    \end{tabular}
    \end{center}
\end{table}
Table \ref{Table:Powercomparisonuniformity} presents the power estimates of various tests for two different alternatives and $\alpha$ values. For the alternatives $C_{1.5}$ and $C_2$, $T_n$ demonstrates higher power estimates compared to all other existing tests, for both $\alpha=0.01$ and $\alpha=0.05$, across all sample sizes ($n$) included in the study. 

\section{Dynamic extropy inaccuracy and cumulative extropy divergence measures}
In many areas such as reliability, survival analysis, economics, etc., the length of time during a study and information about the residual lifetime are essential. In such situations, the information measures are functions of time and thus they are dynamic. In this section, we define the dynamic version of survival extropy inaccuracy and divergence and study their important properties.
\begin{definition}
Let $X_{t} = (X - t|X \geq t)$ and $Y_{t} = (Y|Y \geq t)$  be the residual life random variables of $X$ and $Y$, respectively. Then the dynamic survival extropy inaccuracy (DSEI) between $X_{t}$ and $Y_{t}$ is defined as
\begin{equation}\label{Eq:R3dynamicinaccuracy}
   \xi J_s(\bar{F},\bar{G},t)=-\frac{1}{2}\int_t^\infty \frac{\bar{F}(x)}{\bar{F}(t)}\frac{\bar{G}(x)}{\bar{G}(t)}dx. 
\end{equation}    
\end{definition}
$ \xi J_s(\bar{F},\bar{G},t)$ is always negative and is equal to $J_s(X,t)$ when $\bar{F}(\cdot)=\bar{G}(\cdot)$.
We can write the dynamic survival extropy inaccuracy in terms of survival extropy inaccuracy as follows.
\begin{equation*}
    \begin{split}
     \xi J_s(\bar{F},\bar{G},t)&=-\frac{1}{2}\int_0^\infty \frac{\bar{F}(x)}{\bar{F}(t)}\frac{\bar{G}(x)}{\bar{G}(t)}dx+\frac{1}{2}\int_0^t \frac{\bar{F}(x)}{\bar{F}(t)}\frac{\bar{G}(x)}{\bar{G}(t)}dx. \\ 
     &=\frac{1}{\bar{F}_X(t)\bar{G}_Y(t)}(\xi J_s(X,Y)+\frac{1}{2}\int_0^t \bar{F}_X(x)\bar{G}_Y(x)dx)\\
     &=p(t)(\xi J_s(X,Y)+q(t)),
    \end{split}
\end{equation*}
where $p(t)=\frac{1}{\bar{F}_X(t)\bar{G}_Y(t)}$ and $q(t)=\frac{1}{2}\int_0^t \bar{F}_X(x)\bar{G}_Y(x)dx)$.
The following theorem gives the differential equation of DSEI.
\begin{theorem}
    $\xi J_s(\bar{F},\bar{G},t)$ satisfies the following differential equation.
    \begin{equation}\label{inaccuracydiff}
        \frac{d}{dt}\xi J_s(\bar{F},\bar{G},t)=(h_X(t)+h_Y(t))\xi J_s(\bar{F},\bar{G},t)+\frac{1}{2},
    \end{equation}
    where $h_X (t) = -\frac{\partial}{\partial t} \log \bar{F}(t)$ and $h_Y (t) = -\frac{\partial}{\partial t} \log \bar{G}(t)$ respectively denote the hazard rates of $X$ and $Y$.
\end{theorem}
\begin{proof}
Differentiating (\ref{Eq:R3dynamicinaccuracy}), we get
\[
\frac{d}{dt}\xi J_s(\bar{F},\bar{G},t)=-\frac{1}{2}\left(\frac{-f(t)\bar{G}(t)-g(t\bar{F}(t)}{(\bar{F}(t)\bar{G}(t))^2}
\int_t^\infty\bar{F}(x)\bar{G}(x)dx+1\right)
\]
Replacing $f(t)/\bar{F}(t)$ by $h_X(t)$, $g(t)/\bar{G}(t)$ by $h_Y(t)$ and substituting $\xi J(\bar{F},\bar{G},t)$, we get the equation (\ref{Eq:R3dynamicinaccuracy}).
\end{proof}
\begin{theorem}
    DSEI is nondecreasing (nonincreasing) if and only if 
    \begin{equation}
        \xi J_s(\bar{F},\bar{G},t)\geq(\leq)\frac{-1}{2(h_X(t)+h_Y(t))}
    \end{equation}
\end{theorem}
The following theorem provides a characterization for exponential distribution using the dynamic extropy inaccuracy measure.
\begin{theorem}\label{R3chartheorem1}
    Let $X$ follow an exponential distribution. $\xi J(\bar{F},\bar{G},t)$ is a constant if and only if $Y$ is exponential.
\end{theorem}
\begin{proof}
 Let $X$ follow an exponential distribution so that $h_X(t)=\theta$, a constant. $\xi J_s(\bar{F},\bar{G},t)=c$, is a constant. Substituting in \eqref{inaccuracydiff}, we get 
 \[
 (\theta+h_Y(t))c=-\frac{1}{2}
 \]
 and
 \[
 h_Y(t)=-\frac{1}{2c}-\theta,
 \]
 implies $h_Y(t)$ is a constant which furthur implies $Y$ is exponential. The converse part is proved in the following example.
\end{proof}
\begin{example}
    Let $X$ and $Y$ follow exponential distributions with mean $1/\lambda_1$ and $1/\lambda_2$ respectively. $\xi J_s(\bar{F},\bar{G},t)=-1/2(\lambda_1+\lambda_2)$.
\end{example}
Next, we define the dynamic survival extropy divergence.
\begin{definition}
Let $X_{t}$ and $Y_{t}$  be the residual life random variables of $X$ and $Y$, respectively. Then the dynamic survival extropy divergence(DSED) between $X_{t}$ and $Y_{t}$ is defined as
\begin{equation}
    SJ_r(\bar{F_t}|\bar{G_t})=\frac{1}{2}\int_t^\infty \left(\frac{\bar{F}(x)}{\bar{F}(t)}-\frac{\bar{G}(x)}{\bar{G}(t)}\right)\frac{\bar{F}(x)}{\bar{F}(t)}dx. 
\end{equation}
\end{definition}
In terms of dynamic survival extropy and DSED, we can write it as follows.
\begin{equation}\label{eq:R3relation2}
SJ_r(\bar{F_t}|\bar{G_t})=\xi J_s(\bar{F},\bar{G},t)-J_s(X;t)    
\end{equation}
The following theorem gives the relationship of $SJ_r(\bar{F_t}|\bar{G_t})$ with hazard rates and dynamic survival extropy.
\begin{theorem}
  $SJ_r(\bar{F_t}|\bar{G_t})$ satisfies the differential equation 
  \begin{equation}\label{eq:R3diffeqn1}
     SJ'_r(\bar{F_t}|\bar{G_t})=(h_X(t)+h_Y(t))SJ_r(\bar{F_t}|\bar{G_t})+(h_Y(t)-h_X(t))J_s(X;t).
  \end{equation}
\end{theorem}
\begin{proof}
 Differentiating $SJ_r(\bar{F_t}|\bar{G_t})$, we get
\[
\frac{d}{dt}SJ_r(\bar{F_t}|\bar{G_t})=(h_X(t)+h_Y(t))\xi J_s(\bar{F},\bar{G},t)+\frac{1}{2}-2J_s(X;t)h_X(t)-\frac{1}{2}
\]
Applying \ref{eq:R3relation2}, we get
\[
SJ'_r(\bar{F}_t|\bar{G}_t)=(h_X(t)+h_Y(t))SJ_r(\bar{F}_t|\bar{G}_t)+(h_Y(t)-h_X(t))J_s(X;t).
\]
\end{proof}
\begin{theorem} \label{Theorem:R3divmonotone}
    DSED is nondecreasing (nonincreasing)  if and only if
    \begin{equation}
        SJ_r(\bar{F}_t|\bar{G}_t)\geq(\leq)\frac{h_X(t)-h_Y(t)}{h_X(t)+h_Y(t)}J_s(X;t)
    \end{equation}
\end{theorem}
The following theorem is a characterization of exponential distribution using dynamic survival extropy divergence.
\begin{theorem}
    Let $X$ follow an exponential distribution. Then $SJ_r(\bar{F_t}|\bar{G_t})$ is a constant if and only if $Y$ is exponential.
\end{theorem}
\begin{proof}
Assume that $SJ_r(\bar{F_t}|\bar{G_t})=c$, a constant. If $X$ follows an exponential distribution with mean $1/\lambda$, then $J_s(X;t)=-1/4\lambda$ and $h_X(t)=\lambda$ are independent of $t$. Substituting in (\ref{eq:R3diffeqn1}), we have 
\begin{equation*}
    \begin{split}
        (\lambda+h_Y(t))c+(h_Y(t)-\lambda)(\frac{-1}{4\lambda})&=0\\
        c\lambda+\frac{\lambda}{4\lambda}=\left(\frac{1}{4\lambda}-c\right)h_Y(t),\\
    \end{split}
\end{equation*}
which implies the $h_Y(t)$ is a constant. Or we can prove it directly from (\ref{eq:R3relation2}). Differentiating (\ref{eq:R3relation2}), we have
    \[
    \frac{d}{dt}SJ_r(\bar{F}_t|\bar{G}_t)=\frac{d}{dt}\xi J_s(\bar{F}_t,\bar{G}_t,t)=0
    \]
    which implies $\bar{\xi}J(\bar{F}_t$, $\bar{G}_t,t)$ is a constant. By Theorem \ref{R3chartheorem1}, if $X$ is exponential, then $SJ_r(\bar{F}_t|\bar{G}_t)$ is a constant if and only if $Y$ is exponential. So, $Y$ is exponential. The converse part of the theorem is proved in the Example \ref{EXample:R3dynamicexp}.  
\end{proof}
\begin{example}\label{EXample:R3dynamicexp}
    Let $X$ and $Y$ follow exponential distributions with mean $1/\lambda_1$ and $1/\lambda_2$ respectively. Then $J_s(\bar{F})=-1/4\lambda_1$, $\xi J_s(\bar{F}_t,\bar{G}_t,t)=-1/2(\lambda_1+\lambda_2)$ and then 
    \[
    SJ_r(\bar{F}_t|\bar{G}_t)=\frac{1}{4\lambda_1}-\frac{1}{2(\lambda_1+\lambda_2)}.
    \]
\end{example}
Now, we discuss some results for stochastic ordering.
 The random variable $ X $ is said to be smaller (larger) than or equal to $ Y $ in the 
\begin{enumerate}[(a)]
	\item  Hazard rate ordering, denoted by $ X \leq _{hr}(\geq _{hr})  Y $, if $ h_F(x) \geq (\leq) h_G (x) $ for all $ x \geq 0 $,
	\item Survival extropy ordering, denoted by $ X \leq _{sx} (\geq _{sx}) Y $, if $ J_s(X) \leq (\geq) J_s(Y) $.
\end{enumerate}
Since $ SJ_r(\bar{F}_t|\bar{G}_t)$ lies between $-\infty$ and $\infty$, we derive the conditions of $ SJ_r(\bar{F}_t|\bar{G}_t)$ to be strictly positive and strictly negative in the Theorem \ref{Theorem:R3divalwayspositive} and \ref{Theorem:R3divalwaysnegative} respectively.
\begin{theorem}\label{Theorem:R3divalwayspositive}
    Let $ SJ_r(\bar{F}_t|\bar{G}_t)$ is an increasing function of $t$. If $h_X(t)<h_Y(t)$, then $SJ_r(\bar{F}_t|\bar{G}_t)$ is always positive.
\end{theorem}
\begin{proof}
    See Theorem \ref{Theorem:R3divmonotone}. Let $ SJ_r(\bar{F}_t|\bar{G}_t)$ is strictly increasing and $h_Y(t)>h_X(t)$, then 
    \[
    (h_X(t)-h_Y(t))J_s(X;t)>0.
    \]
    So, we have 
    \[
    SJ_r(\bar{F}_t|\bar{G}_t)>\frac{h_X(t)-h_Y(t)}{h_X(t)+h_Y(t)}J_s(X;t)>0.
    \]
\end{proof}
\begin{corollary}
    Consider an additive hazard model $h_Y(t)=h_X(t)+\theta, \theta \geq 0$. If $SJ_r(\bar{F}_t|\bar{G}_t)$ is an increasing function of $t$, then $SJ_r(\bar{F}_t|\bar{G}_t)$ is always non-negative.
\end{corollary}
\begin{theorem}\label{Theorem:R3divalwaysnegative}
    Let $ SJ_r(\bar{F}_t|\bar{G}_t)$ is a decreasing function of $t$. If $h_X(t)>h_Y(t)$, then $SJ_r(\bar{F}_t|\bar{G}_t)$ is always negative. We have
    \[
    SJ_r(\bar{F}_t|\bar{G}_t)<\frac{h_X(t)-h_Y(t)}{h_X(t)+h_Y(t)}J_s(X;t)<0.
    \]
\end{theorem}
Now, we derive the relationship between $ SJ_r(\bar{F}_t|\bar{G}_t)$ and $ SJ_r(\bar{G}_t|\bar{F}_t)$ in hazard rate ordering.
\begin{theorem}
    If $X\geq^{hr}Y$, then $SJ_r(\bar{F_t}|\bar{G_t})\geq SJ_r(\bar{G_t}|\bar{F_t})$.
\end{theorem}
\begin{proof}
    If $X\geq^{hr}Y$, then $J_s(X;t)\leq J_s(Y;t)$. Since $SJ_r(\bar{F}_t|\bar{G}_t)=\xi J(\bar{F},\bar{G},t)-J_s(X;t)$ and $SJ_r(\bar{G}_t|\bar{F}_t)=\xi J(\bar{F},\bar{G},t)-J_s(Y;t)$, we have if $X\geq^{hr}Y$, then
    \[
    SJ_r(\bar{F}_t|\bar{G}_t)\geq SJ_r(\bar{G}_t|\bar{F}_t).
    \]
\end{proof}
For an additive dynamic survival extropy model and survival extropy ordering, we have the following results. The proofs are excluded since they are trivial from (\ref{eq:R3relation2}).
\begin{theorem}
    $J_s(Y;t)\geq(\leq)J_s(X;t)$ for every $t\geq 0 \iff SJ_r(\bar{F}_t|\bar{G}_t)\geq(\leq)J(\bar{G}_t|\bar{F}_t)$.
\end{theorem}
\begin{theorem}
    $J_s(Y;t)=J_s(X;t)+c$, for any $c>0$ $\iff$ $SJ_r(\bar{F_t}|\bar{G_t})=SJ_r(\bar{G_t}|\bar{F_t})+c$
\end{theorem}
We have defined a symmetric divergence measure based on survival extropy divergence in section 3. We can extend the definition to dynamic divergence measures. The Symmetric dynamic survival extropy divergence is defined as follows.
\begin{definition}
  Let $X$ and $Y$ be two non-negative continuous random variables with survival functions $\bar{F}$ and $\bar{G}$ respectively. Let $SJ_r(\bar{F}_t|\bar{G}_t)$ and $SJ_r(\bar{G}_t|\bar{F}_t)$ be the corresponding dynamic survival extropy divergence measures. Then the symmetric dynamic survival extropy divergence measure between the distributions is defined as
    \begin{equation}\label{Equation:R3symdynamicdiv}
SSJ_r(\bar{F}_t,\bar{G}_t)=\frac{1}{2}\left(SJ_r\left(\bar{F}|\frac{\bar{F}+\bar{G}}{2}\right)+SJ_r\left(\bar{G}|\frac{\bar{F}+\bar{G}}{2}\right)\right).
    \end{equation}     
\end{definition}
Since we have 
\begin{equation}\label{Equation:R3dynasurextdivpart1}
SJ_r\left(\bar{F}_t|\frac{\bar{F}_t+\bar{G}_t}{2}\right)=\frac{1}{4}SJ_r(\bar{F}_t|\bar{G}_t)
\end{equation}
and
\begin{equation}\label{Equation:R3dynasurextdivpart2}
SJ_r\left(\bar{G}_t|\frac{\bar{F}_t+\bar{G}_t}{2}\right)=\frac{1}{4}SJ_r(\bar{G}_t|\bar{F}_t).
\end{equation}
Substituting (\ref{Equation:R3dynasurextdivpart1}) and (\ref{Equation:R3dynasurextdivpart2}) in (\ref{Equation:R3symdynamicdiv}), we rewrite the symmetric survival extropy divergence, $SSJ(\bar{F},\bar{G})$ as
\begin{equation}
 SSJ(\bar{F}_t,\bar{G}_t)=\frac{1}{8}\left( SJ_r(\bar{F}_t|\bar{G}_t)+SJ_r(\bar{G}_t|\bar{F}_t)\right).   
\end{equation}

\section{Estimation and simulation studies of the cumulative extropy divergence}
In this section, we develop a nonparametric estimator for the cumulative extropy divergence. Let $X_{1}$, $X_{2}$, ...$X_{n}$ be t the iid random variables with a common distribution function $F(x)$, and let $Y_{1}$, $Y_{2}$, ...$Y_{n}$ be the iid random variables with a common distribution function $G(y)$. Let $\hat{F}_{n}(x_{(i)})$ and $\hat{G}_{n}(y_{(i)})$ be the empirical distribution functions of $X$ and $Y$, respectively. Then the estimator of $SJ(\bar{F},\bar{G})$ is given by 
\begin{equation}\label{5.1}
    S\hat{J}(\bar{F}|\bar{G})=\frac{1}{2}\sum_{i=1}^{n-1}(\bar{F}_{n}(x_{(i)})-\bar{G}_n(x_{(i)}))\bar{F}_n(x_{(i)})(x_{(i+1)}-x_{(i)})
\end{equation}
Corresponding dynamic cumulative extropy for residual life random variables is
\begin{equation}\label{5.2}
    S\hat{J}(\bar{F_t}|\bar{G_t})=\frac{1}{2}\sum_{i=1}^{n-1}\left(\frac{\bar{F}_{n}(x_{(i)})}{\bar{F}_{n}(t)}-\frac{\bar{G}_n(x_{(i)})}{\bar{G}_{n}(t)}\right)\frac{\bar{F}_{n}(x_{(i)})}{\bar{F}_{n}(t)}(x_{(i+1)}-x_{(i)})
\end{equation}

\begin{table}[h!]
\begin{center}
\caption{Bias and MSE of $SJ_r(\bar{F}|\bar{G})=-0.016667$ using $ F $  and $ G $ are exponential distribution functions with parameters $\lambda_1=3$ and $\lambda_2 =2$ respectively.}
\label{Table:R3estimationofdiv} 
 \begin{tabular}{p{2cm} p{2cm} p{2cm} p{3cm}}
 \hline
 $n$ & $S\hat{J}(\bar{F}|\bar{G})$ & $Bias$ & $MSE$\\ 
 \hline
  30  & -0.019019 & 0.002352 & 0.0001273 \\ 
  50 & -0.017993 & 0.00132 & $7.599381\times10^{-5}$\\
  75 & -0.017665 & 0.000998  & $4.844051\times10^{-5}$\\
 100 & -0.017414 & 0.000747 & $3.63911\times 10^{-5}$\\
 200 & -0.017029 & 0.0003624 & $1.791541\times 10^{-5}$\\
 \hline
\end{tabular}
\end{center}
\end{table}
Table \ref{Table:R3estimationofdiv} shows the bias and MSE of survival extropy divergence estimate for two exponential distributions. It is clear that as $n$ increases, both bias and MSE decrease. 
\begin{table}[h!]
\begin{center}
\caption{Bias and MSE of $SJ_r(\bar{F}_t|\bar{G}_t)=0.125$ using $ F $  and $ G $ are exponential distribution functions with parameters $\lambda_1=1$ and $\lambda_2 =3$ respectively and $t=0.5$.}
\label{Table:R3estimationofdynamicdiv} 
 \begin{tabular}{p{2cm} p{2cm} p{2cm} p{2cm}}
 \hline
 $n$ & $S\hat{J}_r(\bar{F}_t|\bar{G}_t)$ & $Bias$ & $MSE$ \\ 
 \hline
  50  & 0.100400 & 0.024649 & 0.004970 \\ 
  75 & 0.115652 & 0.009348 & 0.002018\\
  100 & 0.122000 & 0.003512  & 0.001507\\
  200 & 0.123097 & 0.001903 & 0.0007652\\
 500 & 0.124602 & 0.0003984 & 0.0003101\\
 \hline
\end{tabular}
\end{center}
\end{table}
Table \ref{Table:R3estimationofdynamicdiv} shows the estimator values of the dynamic survival extropy divergence, its bias, and MSE for different values of $n$ for $t=0.5$. The estimator performs well since the bias and MSE are low and decrease when the sample size increases.\\
\begin{table}[h!]
\begin{center}
\caption{Bias and MSE of $SJ_r(\bar{F}_t|\bar{G}_t)=-0.0083672$ using $ F $   and $ G $ are Gombertz distribution functions with parameters $a=5$ and $b=3$ respectively.}
\label{Table:R3estimationofdynamicdivgombertz} 
 \begin{tabular}{p{1cm} p{1cm} p{2.2cm} p{2.2cm} p{3.5cm}}
 \hline
 $n$ & $t$ & $S\hat{J}_r(\bar{F}|\bar{G})$ & $Bias$ & $MSE$\\ 
 \hline
  50  & 0.1 & -0.01056341 & 0.001045062 & $2.750959\times10^{-5}$ \\ 
  75 & &-0.0101865 & 0.0006682& $1.718583\times 10^{-5}$\\
  100 & &-0.0100108 & 0.000492454  & $1.248892\times10^{-5}$\\
  200& &-0.0097583 & 0.0002399 & $5.899297\times 10^{-6}$\\
  500& &-0.0096278 & 0.000109472 & $2.203666\times 10^{-6}$\\
 \hline
  50  & 0.25 & -0.01106831 & 0.002701143 & $7.915334\times10^{-5}$ \\ 
  75 & &-0.010044 & 0.001677& $3.300782\times 10^{-5}$\\
  100 & &-0.009664 & 0.001297  & $2.165322\times10^{-5}$\\
  200& &-0.0089661 & 0.0005989 & $8.861023\times 10^{-6}$\\
  500& &-0.0086279 & 0.0002608068 & $3.329426\times 10^{-6}$\\
  \hline
   50  & 0.3 & -0.011738 & 0.003727273 & 0.001479221 \\ 
  75 & &-0.0098315 & 0.001820453 & $5.906803\times 10^{-5}$\\
  100 & &-0.0092817 & 0.001270715  & $3.151929\times10^{-5}$\\
  200& &-0.0086947 & 0.000683663 & $1.229168\times 10^{-5}$\\
  500& &-0.00829902 & 0.00028796 & $4.269108\times 10^{-6}$\\
  \hline
\end{tabular}
\end{center}
\end{table}

We also considered two Gompertz distributions for further validation and computed their discrepancy based on the dynamic survival extropy divergence. Table \ref{Table:R3estimationofdynamicdivgombertz} shows the bias and MSE of dynamic survival extropy divergence estimate for different values of $t$ and $n$. As $n$ increases, the bias and MSE decrease asymptotically validate the performance of the estimator.
\section{Application in real data}
The dataset \cite{Blackblaze} includes hard disk properties in the system such as total disk capacity, usage details, failure figures, and daily drive status information. Our subject of interest is the calculated lifetime of each unique hard disk and the lifetimes of the collected hard disk is given in number of days. Using the symmetric survival extropy divergence and symmetric dynamic survival extropy divergence, we compare the lifetime distributions of hard disks with varying capacities. For a comparative study, we first categorize the failure times of hard disks into four groups, say A, B, C, and D based on their capacity. Group A includes hard disks with capacities of 80 GB, 160 GB, 250 GB, 137 GB, 320 GB, and 500 GB. Group B comprises disks with capacities of 1 TB, 1.5 TB, 2 TB, and 3 TB. Group C consists of disks with capacities of 4 TB, 5 TB, and 6 TB. Finally, Group D contains hard disks with capacities of 8 TB, 10 TB, and 12 TB.
\begin{table}[h!]
\begin{center}
\caption{Symmetric Survival extropy divergence for different combinations of groups}
\label{Table:R3harddiskdataanalysis} 
 \begin{tabular}{p{2cm}|p{2cm}|p{2cm}|p{2cm}|p{2cm}}
 \hline
 &A & B & C&D\\ 
 \hline
A&0&0.2895&-88.367&-0.2531\\
B&0.2895&0&-100.96&-0.4733\\
C&-88.367&-100.96&0&-56.4956\\
D&-0.2531&-0.4733&-56.4956&0\\
 \hline
\end{tabular}
\end{center}
\end{table}

We estimate the symmetric cumulative extropy divergence using \eqref{5.1} between the lifetime distributions of all possible combinations of groups A, B, C and D and the estimated values are given in Table \ref{Table:R3harddiskdataanalysis}.  From Table \ref{Table:R3harddiskdataanalysis}, it is obvious that the divergence between the lifetime distributions of all possible combinations of groups A, B and D are considerably lower than that of group C. For a non-negative information metric, we compute the symmetric dynamic survival extropy divergence \eqref{5.2} between hard disks of various capacities for different values of $t$. 
To study the failure time behaviour of different capacities of hard disks, we have then computed the dynamic survival extropy divergence for $t=50$ and $t=100$ days (see Tables \ref{Table:R3harddiskdynamict=50} and \ref{Table:R3harddiskdynamict=100}). Similar to Table \ref{Table:R3harddiskdataanalysis}, for $t=50$ and $t=100$, the divergence between group C and all other groups are relatively higher than all possible combinations of groups A, B, and D. 
\begin{table}[h!]
\begin{center}
\caption{Dynamic survival extropy divergence for different combinations of groups}
\label{Table:R3harddiskdynamict=50} 
\begin{tabular}{p{2cm}|p{2cm}|p{2cm}|p{2cm}|p{2cm}}
 \hline
$t=50$ &A & B & C&D\\ 
 \hline
A&0&0.4384&3.6289&0.3053\\
B&0.4384&0&2.92905&0.04975\\
C&3.6289&2.92905&0&3.224\\
D&0.3053&0.04975&3.224&0\\
 \hline
\end{tabular}
\end{center}
\end{table}

\begin{table}[h!]
\begin{center}
\caption{Dynamic survival extropy divergence for different combinations of groups}
\label{Table:R3harddiskdynamict=100} 
\begin{tabular}{p{2cm}|p{2cm}|p{2cm}|p{2cm}|p{2cm}}
 \hline
$t=100$ &A & B & C&D\\ 
 \hline
A&0&0.4903&2.6251&0.4123\\
B&0.4903&0&2.7366&0.0503\\
C&2.6251&2.7366&0&2.8945\\
D&0.4123&0.0503&2.8945&0\\
 \hline
\end{tabular}
\end{center}
\end{table}
\begin{table}[h!]
\begin{center}
\caption{Dynamic survival extropy divergence for different combinations of groups}
\label{Table:R3harddiskdynamict=250} 
\begin{tabular}{p{2cm}|p{2cm}|p{2cm}|p{2cm}|p{2cm}}
 \hline
$t=150$ &A & B & C&D\\ 
 \hline
A&0&0.9664&1.0313&0.8289\\
B&0.9664&0&2.1344&0.1636\\
C&1.0313&2.1344&0&2.4403\\
D&0.8289&0.1636&2.4403&0\\
 \hline
\end{tabular}
\end{center}
\end{table}
We have also estimated the dynamic extropy divergence for $t=150$ days, given in Table \ref{Table:R3harddiskdynamict=250}. 
Comparing the tables, it is evident that the cumulative extropy divergence between the failure time distribution of Group C have higher divergence values compared to all other groups. Using Tables \ref{Table:R3harddiskdynamict=50} -\ref{Table:R3harddiskdynamict=250}, one can easily infer that the failure times pattern of Groups A, B, and D are relatively similar, which are different from Group C, or the failure times pattern of hard disk capacities under Group C stands different from others.
\section{Concluding Remarks}
The extropy-based divergence measures are relatively new in information theory and we have introduced a cumulative inaccuracy and divergence measures based on survival extropy and studied some of their applications. We have defined a survival extropy inaccuracy ratio using the proposed inaccuracy measures, and examined its application in image analysis. We have estimated the survival extropy inaccuracy ratio between images and classified the images based on the magnitude of the estimates. We have also introduced the survival extropy divergence measure and studied its dynamic behavior. Furthermore, we have developed a test for assessing goodness of fit for uniform $U(0,b)$ distributions using the survival extropy divergence measure. We have validated and compared the performance of the test with some important tests available in the literature. Characterizations for exponential distribution using both measures were derived. Also, we have estimated the divergence measure using a non-parametric estimation method and validated its performance using simulation studies.  Finally, the article has demonstrated the significance of the estimators through real data sets relevant in survival studies.

\section*{Acknowledgements}

The first and second authors thank Cochin University of Science and Technology, India, for their financial support.

\section*{Conflict of Interest statement}
On behalf of all the authors, the corresponding author states that there is no conflict of interest.
\section*{Data availability statement}
 The Chinese MNIST dataset and Hard drive dataset that support the findings of this study are openly available in reference \cite{Blackblaze} and in reference, \cite{kaggle_chinese_mnist} respectively. 
\bibliographystyle{apalike}
\bibliography{main}
\end{document}